\documentclass[a4paper,UKenglish,cleveref,autoref,thm-restate]%
{lipics-v2021}
\pdfoutput=1

\usepackage{mathtools}

\title{Linear Termination over $\pmb{\NN}$ is Undecidable}
\titlerunning{Linear Termination over $\pmb{\NN}$ is Undecidable}

\author{Fabian Mitterwallner}
{University of Innsbruck, Innsbruck, Austria}
{fabian.mitterwallner@uibk.ac.at}
{https://orcid.org/0000-0001-5992-9517}
{}
\author{Aart Middeldorp}
{University of Innsbruck, Innsbruck, Austria}
{aart.middeldorp@uibk.ac.at}
{https://orcid.org/0000-0001-7366-8464}
{}
\author{Ren\'{e} Thiemann}
{University of Innsbruck, Innsbruck, Austria}
{rene.thiemann@uibk.ac.at}
{https://orcid.org/0000-0002-0323-8829}
{}

\authorrunning{F. Mitterwallner and A. Middeldorp and R. Thiemann}

\Copyright{Fabian Mitterwallner and Aart Middeldorp and Ren\'e Thiemann}



\category{}

\acknowledgements{We thank Nao Hirokawa for posing the question about
linear interpretations.}

\EventEditors{John Q. Open and Joan R. Access}
\EventNoEds{2}
\EventLongTitle{42nd Conference on Very Important Topics (CVIT 2016)}
\EventShortTitle{CVIT 2016}
\EventAcronym{CVIT}
\EventYear{2016}
\EventDate{December 24--27, 2016}
\EventLocation{Little Whinging, United Kingdom}
\EventLogo{}
\SeriesVolume{42}
\ArticleNo{23}
\newcommand{\h}[1][.3]{\hspace{#1mm}}
\newcommand{\m}[1]{\mathsf{#1}}
\newcommand{\x}[1]{\mathcal{#1}}
\newcommand{\seq}[2][n]{#2_1,\dotsc,#2_{#1}}
\newcommand{\SET}[1]{\{\h#1\h\}}
\newcommand{\set}[1][n]{\SET{1,\dotsc,#1}}
\newcommand{\MC}[2][0]{\makebox[#1mm]{#2}}

\newcommand{\xF}{\x{F}}
\newcommand{\xQ}{\x{Q}}
\newcommand{\xR}{\x{R}}

\newcommand{\mI}[2][\NN]{\m{#2}_{#1}}

\newcommand{\NN}{\mathbb{N}}
\newcommand{\ZZ}{\mathbb{Z}}
\newcommand{\QQ}{\mathbb{Q}}
\newcommand{\RR}{\mathbb{R}}
\newcommand{\DD}{\mathcal{D}}

\newcommand\fz{\m{z}}

\newcommand\fg{\m{g}}

\newcommand\fa{\m{a}}

\newcommand\fq{\m{q}}
\newcommand\fh{\m{h}}

\newcommand\fv[1]{\m{v}_{#1}}

\newcommand{\I}[2][\NN]{{#2}_{#1}}
\newcommand{\vI}[2][\NN]{[#2]_{#1}}

\newcommand{\ceil}[1]{{\ulcorner #1 \urcorner}}

\newcommand{\enc}[1]{\ceil{#1}}

\newcommand{\sem}[2][\NN]{\vI[#1]{#2}}

\newcommand{\nR}[1]{\,\xrightarrow{\,\MC[2]{$\scriptstyle #1$}\,}\,}

\newcommand{\lemref}[1]{Lemma~\ref{lem:#1}}
\newcommand{\thmref}[1]{Theorem~\ref{thm:#1}}

\nolinenumbers
\hideLIPIcs

\begin{document}

\maketitle

\begin{abstract}
Recently it was shown that it is undecidable whether a term rewrite system
can be proved terminating by a polynomial interpretation in the
natural numbers. In this paper we show that this is also the case when
restricting the interpretations to linear polynomials, as is often done in
tools, and when only considering single-rule rewrite systems. What is more,
the new undecidability proof is simpler than the previous one. We further
show that polynomial termination over the rationals/reals is undecidable.
\end{abstract}

\section{Introduction}
\label{sec:introduction}

In a recent paper~\cite{MM22} the problem of whether
a finite term rewrite system (TRS) can be proved terminating by a polynomial
interpretation in the naturals numbers was shown to be undecidable. The
result was strengthened by restricting the instance to incremental
polynomially
terminating TRSs. Moreover, incremental polynomial termination over $\NN$ is
an undecidable property of terminating TRSs.

In this paper we complement these results by proving the somewhat surprising
result that the problem remains undecidable if we restrict the allowed
interpretation functions to linear ones, even for single-rule polynomially
terminating TRSs. A second contribution is that the termination problem is
undecidable if we consider polynomial interpretations over the rationals
and reals. Our undecidability proofs are surprisingly simple.

The results in this paper are obtained by a reduction
from a variant of Hilbert's tenth problem. Hilbert's tenth is one of 23
problems published by David Hilbert in 1900 which where all unsolved at
the time. To
solve the tenth problem one should find an algorithm that given a
diophantine equation with integer coefficients determines if the equation
has a solution in the integers~\cite{H02}. As is turns out this is
impossible and the underlying decision problem was proved undecidable by
Matijasevic in 1970~\cite{M70}.

To simplify the encoding of Hilbert's tenth problem, we first reduce it to
a slightly modified decision problem. Instead of using an arbitrary integer
polynomial, we consider two polynomials $P$ and $Q$ with only positive
coefficients and ask if $P(\seq{x}) \geqslant Q(\seq{x})$ for some arguments
$\seq{x} \in \NN_+$. This is also undecidable and is
more easily applicable in the proofs related to polynomial termination.

\begin{lemma}
\label{lem:DP}
The following decision problem is undecidable: \\
\begin{tabular}[t]{@{\quad}ll}
instance: & two polynomials $P$ and $Q$ with positive integer
coefficients \\
question: & $P(\seq{x}) \geqslant Q(\seq{x})$ for some $\seq{x} \in \NN_+$ ?
\end{tabular}
\end{lemma}

\begin{proof}
We proceed by a reduction from Hilbert's 10th problem.
Assume the decision problem is decidable and let $R \in \ZZ[\seq{x}]$
be a polynomial. We can modify Hilbert's 10th problem for $R$ as follows:
\begin{align*}
&\exists\,\seq{x} \in \ZZ~R(\seq{x}) = 0 \\
\iff~&\exists\,\seq{x} \in \ZZ~R(\seq{x})^2 \leqslant 0 \\
\iff~&\exists\,\seq{a} \in \SET{-1,0,1}~\exists\,\seq{x} \in \NN_+~
R(a_1x_1,\dots,a_nx_n)^2 \leqslant 0
\tag{1}
\intertext{%
We can now split $R(a_1x_1,\dots,a_nx_n)^2$ into two polynomials
$P_{\vec a}$ and $Q_{\vec a}$ containing only positive coefficients, such
that
$R(a_1x_1,\dots,a_nx_n)^2 = Q_{\vec a}(\seq{x}) - P_{\vec a}(\seq{x})$.
Hence (1) is equivalent to
}
&\exists\,\seq{a} \in \SET{-1,0,1}~\exists\,\seq{x} \in \NN_+~
P_{\vec a}(\seq{x}) \geqslant Q_{\vec a}(\seq{x})
\end{align*}
The final problem is decidable by our assumption, since it
consists of $3^n$ instances of the decision problem. This contradicts
the undecidability of Hilbert's 10th problem, thereby proving the lemma.
\end{proof}

\section{Undecidability of Linear Termination over $\NN$}
\label{sec:undecidability}

To prove undecidability of linear termination we define a
TRS $\xR$ which is parameterized by two polynomials $P$ and $Q$
containing only positive coefficients. We then prove that $\xR$ can be
shown to be terminating using a linear interpretation if and
only of $P(\seq{x}) \geqslant Q(\seq{x})$ for some $\seq{x} \in \NN_+$.
For polynomials containing the indeterminates $\seq{v}$, the signature of
$\xR$ is $\xF = \SET{\m{z},\m{o},\m{a},\m{f},\fv{1},\dots,\fv{n}}$,
where $\m{z}$ and $\m{o}$ are constants,
$\fv{1},\dots,\fv{n}$ are unary symbols,
$\m{a}$ is a binary symbol and $\m{f}$ has arity four.

To this end we first define an encoding $\ceil{\cdot}^x$, which maps
polynomials with positive coefficients to terms containing the variable
$x$.

\begin{definition}
\label{def:encoding}
Let $P$ be a polynomial containing only positive coefficients, and the
indeterminates $\seq{v}$. We can then encode natural numbers as
\begin{align*}
\ceil{0}^x &= \m{z} & \ceil{m+1}^x &= \m{a}(x,\ceil{m}^x)
\end{align*}
A monomial $M = c \cdot v_1^{m_1} \cdot v_2^{m_2} \cdots v_n^{m_n}$
with $c \in \NN_+$ and $\seq{m} \in \NN$ is encoded as
\begin{align*}
\ceil{M}^x &= \m{v}_{\m{1}}^{m_1}(\m{v}_{\m{2}}^{m_2}( \cdots
(\m{v}_{\m{n}}^{m_n}(\ceil{c}^x)) \cdots ))
\end{align*}
where $v^0(t) = t$ and $v^{i+1}(t) = v(v^i(t))$
for $v \in \SET{\fv{1},\dots,\fv{n}}$. Finally the polynomial
$P = M_1 + M_2 + \cdots + M_k$ is encoded as
\[
\ceil{P}^x = \m{a}(\ceil{M_1}^x,\m{a}(\ceil{M_2}^x, \cdots
\m{a}(\ceil{M_k}^x,\m{z}) \cdots ))
\]
\end{definition}

\begin{example}
For the polynomial $P = X^3 + 2X + 2$ we obtain the term
\[
\ceil{P}^y =
\m{a}(\,\underbrace{\m{X}(\m{X}(\m{X}(\m{a}(y,\m{z}))))}_{\ceil{X^3}}\,,
\m{a}(\,\underbrace{\m{X}(\m{a}(y,\m{a}(y,\m{z})))}_{\ceil{2X}}\,,
\m{a}(\,\underbrace{\m{a}(y,\m{a}(y,\m{z}))}_{\ceil{2}}\,,\m{z})))
\]
\end{example}

The TRS $\xR$ can then be defined via this encoding.

\begin{definition}
For polynomials $P$ and $Q$ containing only positive coefficients we
obtain the \textup{TRS} $\xR$ consisting of the single rule
\[
\m{f}(y_1,y_2,\m{a}(\ceil{P}^{y_3},y_3),\m{o}) \to
\m{f}(\m{a}(y_1,\m{z}),\m{a}(\m{z},y_2),\m{a}(\ceil{Q}^{y_3},y_3),\m{z})
\]
\end{definition}

The rule serves two purposes. First it constrains any linear
interpretation proving its termination to conform to a very limited
shape. Secondly it uses these restricted shapes
to encode the inequality $P \geqslant Q$ in the orientation of the rule
$\vI{\ell} > \vI{r}$. This leads to the following result.

\begin{theorem}
\label{thm:linPT}
Termination of $\xR$ can be shown by a linear interpretation
if and only if $P(\seq{v}) \geqslant Q(\seq{v})$ for some
$\seq{v} \in \NN_+$.
\end{theorem}

\begin{proof}
For the if direction assume $P(\seq{v}) \geqslant Q(\seq{v})$
for some $\seq{v} \in \NN_+$. We then choose the monotone interpretations
\begin{align*}
\mI{z} &= 0 & \mI{a}(x_1,x_2) &= x_1 + x_2 &
\mI{v_{\mathit{i}}}(x) &= v_i \cdot x \quad \text{for all $i \in \set$} \\
\mI{o} &= 1 & \mI{f}(x_1,x_2,x_3,x_4) &= x_1 + x_2 + \mathrlap{x_3 + x_4}
\end{align*}
Note that using this interpretation we have
$\vI{\ceil{P}^{y_3}} = P(\seq{v}) \cdot y_3$,
and the same holds for $Q$. Hence we orient the rule in $\xR$, as
seen in
\[
\vI{\ell} = y_1 + y_2 + (P(\seq{v}) + 1) y_3 + 1 >
y_1 + y_2 + (Q(\seq{v}) + 1) y_3 = \vI{r}
\]
For the only-if direction we assume a linear interpretation for all
$f \in \xF$, such that $\vI{\ell} > \vI{r}$.
Hence we know that all interpretations have the shape
$\I{f}(\seq[k]{x}) = f_0 + f_1 x_1 + \cdots + f_k x_k$
where $f_0 \in \NN$ and $\seq[k]{f} \in \NN_+$ due to monotonicity.
For any term $t$ we write $\vI{t}^{y_i}$ for the coefficient of the
indeterminate $y_i$ of the linear polynomial $\vI{t}$. Using this
notation, $\vI{\ell} > \vI{r}$ implies
$\vI{\ell}^{y_i} \geqslant \vI{r}^{y_i}$ for $i \in \SET{1,2,3}$.
By the shape of the rule we deduce
$\m{f_1} = \vI{\ell}^{y_1} \geqslant \vI{r}^{y_1} = \m{f_1} \m{a_1}$
and in combination with $\m{f_1} > 0$ and $\m{a_1} > 0$ we conclude
$\m{a_1} = 1$.
Similarly, from $\vI{\ell}^{y_2} \geqslant \vI{r}^{y_2}$ we infer
$\m{a_2} = 1$, and in turn
$\mI{a}(x_1,x_2) = x_1 + x_2 + \m{a_0}$ for some $\m{a_0} \in \NN$.
Due to the shape $\mI{a}$ it is easy to see that
$\vI{\ceil{m}^{y_3}}^{y_3} = m$ for any $m \in \NN$,
$\vI{c \cdot \ceil{v_1^{m_1} \cdots v_n^{m_n}}^{y_3}}^{y_3} =
c \cdot v_1^{m_1} \cdots v_n^{m_n}$ and further
$\vI{\ceil{P}^{y_3}}^{y_3} = P(\seq{v})$ for any polynomial $P$. Hence
\[
\m{f_3} \cdot (P(\seq{v}) + 1) = \vI{\ell}^{y_3}
\geqslant \vI{r}^{y_3} = \m{f_3} \cdot (Q(\seq{v}) + 1)
\]
Since $\m{f_3} > 0$, division by $\m{f_3}$ is possible, resulting in the
desired inequality $P(\seq{v}) \geqslant Q(\seq{v})$ for
$\seq{v} \in \NN_+$.
\end{proof}

\begin{corollary}
Linear termination is undecidable, even for single-rule \textup{TRS}s.
\end{corollary}

\begin{proof}
This follows directly from \thmref{linPT} and \lemref{DP}.
\end{proof}

Interestingly the TRS $\xR$ is always terminating, independent of the
polynomials $P$ and $Q$. This can be shown using a (non-linear)
polynomial interpretation.

\begin{lemma}
\label{lem:RisPT}
The \textup{TRS} $\xR$ is polynomially terminating.
\end{lemma}

\begin{proof}
Use the following monotone interpretation over $\NN$
\begin{align*}
\mI{o} &= Q(1,\dots\mathrlap{,1) + 1} &
\mI{a}(x,y) &= x + y &
\mI{v_{\mathit{i}}}(x) &= x \quad \text{for all $i \in \set$} \\
\mI{z} &= 0 &
\mI{f}(x_1,x_2,x_3,x_4) &= x_3x_4 + x_1 \mathrlap{{} + x_2 + x_3 + x_4}
\end{align*}
Note that due to $\vI{\m{v_{\mathit{i}}}(x)}^x = 1$ we have
$\vI{\ceil{P}^{y_3}}^{y_3} = P(1,\dots,1)$ and
$\vI{\ceil{Q}^{y_3}}^{y_3} = Q(1,\dots,1)$.
Hence, we can orient the rule as seen in
\begin{align*}
\vI{\ell} &= (Q(1,\dots,1) + 1) (P(1,\dots,1) + 1) y_3 + {} \\
&\phantom{{} = {}} y_1 + y_2 + (P(1,\dots,1) + 1) y_3 + (Q(1,\dots,1) + 1)
\\
&> y_1 + y_2 + (Q(1,\dots,1) + 1) y_3 = \vI{r} \qedhere
\end{align*}
\end{proof}

\begin{corollary}
Linear termination is undecidable, even for polynomially
terminating single-rule systems.
\end{corollary}

\section{Polynomial Termination over $\QQ$ and $\RR$}
\label{sec:R and Q}

When considering polynomial interpretations over $\QQ$ and $\RR$, we
restrict the domain to all non-negative values. Moreover, when comparing the
polynomials associated with the left- and right-hand side of a rewrite rule,
we demand that the difference is at least $\delta$, for some fixed positive
value $\delta$ of the domain. This ensures termination. We refer to \cite{NM14}
for formal definitions as well as the relationship
between the notions of polynomial termination over $\NN$, $\QQ$ and $\RR$.

In the previous section we encoded polynomials as terms such that
indeterminates of the polynomials correspond to coefficients of some
interpretation.
When dealing with polynomial termination over $\QQ$ and $\RR$
a new approach for proving undecidability is required,
since coefficients take values in $\QQ$ or $\RR$.
However, what does not change is that the exponents of our
interpretations must still be natural numbers. We can make use of this by
encoding the polynomials and the order on polynomials in the degrees of
our interpretations. As long as we can represent multiplication in the
interpretations we can use basic arithmetic to encode the polynomials in
the degrees.

\begin{lemma}
If $P$ and $Q$ are univariate polynomials containing only positive
coefficients then
\begin{enumerate}
\item
$\deg(P + Q) = \max(\deg(P),\deg(Q))$,
\smallskip
\item
$\deg(P \cdot Q) = \deg(P) + \deg(Q)$, and
\smallskip
\item
$\deg(P \circ Q) = \deg(P) \cdot \deg(Q)$. \qed
\end{enumerate}
\end{lemma}

For encoding polynomials with positive coefficients
as terms we use \Cref{def:encoding}, so using function symbols from
$\SET{\fz,\fa} \cup \SET{\fv{i} \mid 1 \leqslant i \leqslant n}$.
Moreover, we write $\enc{P}$ for $\enc{P}^x$ with some fixed variable $x$.
The polynomial is then encoded in the degree of $\enc{P}$, as seen in the
following lemma, which can be proved using a simple induction over
\Cref{def:encoding}.

\begin{lemma}
\label{lem:deg encoding}
Let $\DD \in \SET{\QQ,\RR}$ and suppose $\mI[\DD]{z} = z_0$ and
$\mI[\DD]{a} = a_3 x y + a_2 x + a_1 y + a_0$ for some
$z_0, a_0 \in \DD_{\geqslant 0}$ and $a_3, a_2, a_1 \in \DD_{> 0}$.
If $P \in \ZZ[\seq{v}]$ with positive coefficients then
$\deg(\sem[\DD]{\enc{P}}) =
P(\deg(\sem[\DD]{\fv{1}}),\dots,\deg(\sem[\DD]{\fv{n}}))$.
\end{lemma}

\begin{proof}
We use induction on the definition of $\enc{P}$.
If $P = 0$ then $\sem[\DD]{\enc{P}} = \fz$ and thus
$\deg(\sem[\DD]{\enc{P}}) = 0 = P$. For $P = m+1$ we obtain
$\enc{P} = \m{a}(x,\enc{m})$ and thus
\[
\sem[\DD]{\enc{P}} = a_3 \cdot \sem[\DD]{\enc{m}} \cdot x + a_2 \cdot x +
a_1 \cdot \sem[\DD]{\enc{m}} + a_0
\]
Hence $\deg(\sem[\DD]{\enc{P}}) = \deg(\sem[\DD]{\enc{m}}) + 1 = m + 1$
by the induction hypothesis. For a monomial
$P = c \cdot v_1^{m_1} \cdots v_n^{m_n}$ with
$c \in \NN_+$ and $\seq{m} \in \NN$ we obtain
\begin{align*}
\deg(\sem[\DD]{\enc{P}}) &= \deg(\sem[\DD]{\enc{c}}) \cdot
\deg(\sem[\DD]{\fv{1}})^{m_1}\,\cdots\,\deg(\sem[\DD]{\fv{n}})^{m_n} \\
&= c \cdot
\deg(\sem[\DD]{\fv{1}})^{m_1}\,\cdots\,\deg(\sem[\DD]{\fv{n}})^{m_n} \\
&= P(\deg(\sem[\DD]{\fv{1}}),\dots,\deg(\sem[\DD]{\fv{n}}))
\end{align*}
Finally, if $P = M_1 + \cdots + M_k$ then
$\enc{P} = \m{a}(\enc{M_1}, \cdots \m{a}(\enc{M_k},\m{z}) \cdots )$
and thus
\begin{align*}
\deg{\sem[\DD]{\enc{P}}} &=
\deg(\sem[\DD]{\enc{M_1}}) + \cdots + \deg(\sem[\DD]{\enc{M_k}}) \\
&= M_1(\deg(\sem[\DD]{\fv{1}}),\dots,\deg(\sem[\DD]{\fv{n}})) + \cdots +
M_k(\deg(\sem[\DD]{\fv{1}}),\dots,\deg(\mathrlap{\sem[\DD]{\fv{n}}))} \\
&= P(\deg(\sem[\DD]{\fv{1}}),\dots,\deg(\sem[\DD]{\fv{k}})) \qedhere
\end{align*}
\end{proof}

\begin{definition}
Given polynomials $P$ and $Q$ containing only positive coefficients and
containing the indeterminates $\seq{v}$, the \textup{TRS} $\xQ$ is defined
over the signature
$\xF = \SET{\fz,\fa,\fh,\fq,\fg} \cup \SET{\fv{i} \mid i \in \set}$ and
consists of the rules
\begin{align*}
\fq(\fh(x)) &\nR{1} \fh(\fh(\fq(x))) &
\fa(x,x) &\nR{5} \fq(x) \\
\fh(x) &\nR{2} \fg(x,x) &
\fh(x) &\nR{6} \fa(\fz,x) \\
\fg(\fq(x),\fh(\fh(\fh(x)))) &\nR{3} \fq(\fg(x,\fh(\fz))) &
\fh(x) &\nR{7} \fa(x,\fz) \\
\fh(\fq(x)) &\nR{4} \fa(x,x) &
\fh(\fa(\enc{P},x)) &\nR{8} \fa(\enc{Q},x) % encode p >= q
\end{align*}
\end{definition}

The main idea behind this system is that rules (1) through (7)
restrict the possible interpretations of $\mI[\DD]{a}$ and $\mI[\DD]{h}$
such that \lemref{deg encoding} is applicable, and that compatibility with
rule (8) implies $P(\seq{v}) \geqslant Q(\seq{v})$.
The rules (1)--(7) are similar to ones already used in \cite{NMZ10},
where they restrict possible interpretations over $\NN$,
and in \cite{NM14}, where they are also applied to interpretations over
$\QQ$ and $\RR$.

\begin{theorem}
\label{thm:QR is undecidable}
For $\DD \in \SET{\QQ,\RR}$ and polynomials $P, Q \in \ZZ[\seq{x}]$
with positive integer coefficients the \textup{TRS} $\xQ$ can be
proved terminating
using a polynomial interpretation over $\DD$ if and only if
$P(\seq{v}) \geqslant Q(\seq{v})$ for some $\seq{v} \in \NN_+$.
\end{theorem}

\begin{proof}
For the if direction assume $P(\seq{v}) \geqslant Q(\seq{v})$ for some
$\seq{v} \in \NN_+$. Take the interpretations
\begin{align*}
\mI[\DD]{z} &= 0 &
\mI[\DD]{g}(x,y) &= x + y &
\mI[\DD]{a}(x,y) &= xy + x + y + 1 \\
\mI[\DD]{q}(x) &= x^2 + 2x &
\mI[\DD]{h}(x) &= hx + h &
\mI[\DD]{v_{\mathit{i}}}(x) &= x^{v_i} \quad \text{for all $i \in \set$}
\end{align*}
where $h > 2$. With $\delta = 1$ these orient the rules
(1) -- (7). \lemref{deg encoding} yields
$\deg(\sem[\DD]{\enc{P}}) =
P(\deg(\sem[\DD]{\fv{1}}),\dots,\deg(\sem[\DD]{\fv{n}})) =
P(\seq{v})$ and similarly $\deg(\sem[\DD]{\enc{Q}}) = Q(\seq{v})$.
From the assumption $P(\seq{v}) \geqslant Q(\seq{v})$ we therefore
obtain $\deg(\sem[\DD]{\enc{P}}) \geqslant \deg(\sem[\DD]{\enc{Q}})$.
Consequently,
\[
\deg(\sem[\DD]{\fh(\fa(\enc{P},x))}) = \deg(\sem[\DD]{\enc{P}}) + 1
\geqslant \deg(\sem[\DD]{\enc{Q}}) + 1 = \sem[\DD]{\fa(\enc{Q},x)}
\]
It follows that by choosing the coefficient $h$ large enough,
the remaining rule (8) is oriented.

For the only-if direction assume the TRS $\xQ$ is
polynomially terminating over $\DD$. From compatibility
with rule (1) we infer
\[
\deg(\sem[\DD]{\fq(\fh(x))}) =
\deg(\mI[\DD]{q}) \cdot \deg(\mI[\DD]{h}) \geqslant
\deg(\mI[\DD]{q}) \cdot \deg(\mI[\DD]{h})^2 =
\deg{\sem[\DD]{\fh(\fh(\fq(x)))}}
\]
Hence $\deg(\mI[\DD]{h}) = 1$ and thus $\mI[\DD]{h}(x) = h_1 x + h_0$
for some $h_1 \geqslant 1$ and $h_0 \geqslant 0$. From compatibility
with rule (2) we infer $\deg(\mI[\DD]{g}) = 1$ and thus
$\mI[\DD]{g}(x,y) = g_2 x + g_1 y + g_0$ with $g_2, g_1 \geqslant 1$.
Moreover, $h_1 \geqslant g_1 + g_2 \geqslant 2$ and
$h_0 >_\delta g_0 \geqslant 0$. Looking back at rule (1) we now
can infer that $\mI[\DD]{q}$ is at least quadratic, since if it were
linear we would obtain the inequality
\[
q_1 h_1 \cdot x + q_1 h_0 + q_0 >_\delta
q_1 h_1^2 \cdot x + h_1^2 q_0 + h_1 h_0 + h_0
\]
for all $x \in \DD_{\geqslant 0}$. This can only hold if
$q_1 h_1 \geqslant q_1 h_1^2$, which in turn implies $h_1 \leqslant 1$,
contradicting $h_1 \geqslant 2$. Next we show
$\deg(\mI[{\DD}]{q}) = 2$.
Compatibility with rule (3) induces the constraint
$g_2 \cdot \sem[\DD]{\fq(x)} + g_1 \cdot
\sem[\DD]{\fh(\fh(\fh(x)))} + g_0 >_\delta
\sem[\DD]{\fq(\fg(x,\fh(\fz)))}$, which implies
\[
1 = \deg(\sem[\DD]{\fh(\fh(\fh(x)))} \geqslant
\deg(\sem[\DD]{\fq(\fg(x,\fh(\fz)))} - g_2 \cdot \sem{\fq(x)})
\]
Since $h_0 > 0$ and $\sem[\DD]{\fh(\fz)} > 0$ this can only be the
case if $\deg(\mI[{\DD}]{q}) = 2$. Using this fact together with
compatibility with the rules (4) and (5) we
infer $\deg(\mI[{\DD}]{a}) = 2$ and hence
$\mI[{\DD}]{a}(x,y) = a_5x^2 + a_4y^2 + a_3xy + a_2x + a_1y + a_0$.
Compatibility with rules (6) and (7) implies
$a_5 = a_4 = 0$ resulting in
$\mI[{\DD}]{a}(x,y) = a_3xy + a_2x + a_1y + a_0$ with
$a_3, a_2, a_1 \in \DD_{> 0}$. Compatibility with (8) implies
$\deg(\sem[\DD]{\enc{P}}) + 1 \geqslant \deg(\sem[\DD]{\enc{Q}}) + 1$.
With the help of \lemref{deg encoding} we obtain
$P(\deg(\sem[\DD]{\fv{1}}),\dots,\deg(\sem[\DD]{\fv{n}})) \geqslant
Q(\deg(\sem[\DD]{\fv{1}}),\dots,\deg(\sem[\DD]{\fv{n}}))$.
\end{proof}

\begin{corollary}
Polynomial termination over $\QQ$ and $\RR$ is undecidable. \qed
\end{corollary}

\bibliography{references}

\end{document}